\documentclass{article}

\usepackage{arxiv}

\usepackage[utf8]{inputenc} 
\usepackage[T1]{fontenc}    
\usepackage{hyperref}       
\usepackage{url}            
\usepackage{booktabs}       
\usepackage{amsfonts}       
\usepackage{nicefrac}       
\usepackage{microtype}      
\usepackage{lipsum}
\usepackage{graphicx}
\graphicspath{ {./images/} }
\usepackage{setspace}

\begin{onehalfspacing}
\title{PAC Reinforcement Learning Algorithm for General-Sum Markov Games}
\end{onehalfspacing}

\author{
 Ashkan Zehfroosh \\
  Department of Mechanical Engineering\\
  University of Delaware\\
  Newark, DE 19716 \\
  \texttt{ashkanz@udel.edu} \\
   \And
 Herbert G. Tanner \\
  Department of Mechanical Engineering\\
  University of Delaware\\
  Newark, DE 19716 \\
  \texttt{btanner@udel.edu} \\
}

\usepackage{graphicx}
\usepackage{graphics} 
\usepackage{epsfig} 
\usepackage{color,soul}
\usepackage{amsmath} 
\usepackage{amssymb,mathtools}  
\usepackage{paralist}
\usepackage{enumitem}
\usepackage{acronym}
\usepackage{booktabs}
\usepackage{verbatim}
\usepackage{cite}
\usepackage{balance}
\usepackage{comment}
\usepackage[utf8]{inputenc}
\usepackage[english]{babel}
\usepackage{MnSymbol}
\usepackage{wasysym}
\usepackage{caption}
\usepackage{subcaption}
\usepackage{amsthm}

\usepackage{algorithm}
\usepackage[flushleft]{threeparttable}
\usepackage[normalem]{ulem}
\usepackage{algpseudocode}
\usepackage{varwidth}
\usepackage{xfrac}
\usepackage{hhline}
\usepackage{todonotes}
\usepackage{systeme}
\usepackage{lipsum}
\usepackage{appendix}

\usepackage{lipsum}
\newcommand\blfootnote[1]{%
  \begingroup
  \renewcommand\thefootnote{}\footnote{#1}%
  \addtocounter{footnote}{-1}%
  \endgroup
}

\begin{document}

\setlength{\abovedisplayskip}{9pt}
\setlength{\belowdisplayskip}{9pt}

\blfootnote{This work was supported in part by NIH under R01HD87133 and DTRA under HDTRA1-16-1-0039.}

\acrodef{hri}[\textsc{hri}]{human-robot interaction}
\acrodef{mdp}[\textsc{mdp}]{Markov decision process}
\acrodef{smdp}[\textsc{smdp}]{Semi-Markov decision process}
\acrodef{ai}[\textsc{ai}]{Artificial Intelligence}
\acrodef{ml}[\textsc{ml}]{maximum likelihood}
\acrodef{pomdp}[\textsc{pomdp}]{partially observable Markov decision process}
\acrodef{momdp}[\textsc{momdp}]{mixed observability Markov decision process}
\acrodef{nlp}[\textsc{nlp}]{natural language processing}
\acrodef{pac}[\textsc{pac}]{probably approximately correct}
\acrodef{rl}[\textsc{rl}]{reinforcement learning}
\acrodef{ddq}[\textsc{ddq}]{Dyna-Delayed Q-learning}
\acrodef{tdm}[\textsc{tdm}]{Temporal Difference Models}
\acrodef{hri}[\textsc{hri}]{human-robot interaction}
\acrodef{mdp}[\textsc{mdp}]{Markov decision process}
\acrodef{smdp}[\textsc{smdp}]{Semi-Markov decision process}
\acrodef{ai}[\textsc{ai}]{Artificial Intelligence}
\acrodef{ml}[\textsc{ml}]{maximum likelihood}
\acrodef{pomdp}[\textsc{pomdp}]{partially observable Markov decision process}
\acrodef{momdp}[\textsc{momdp}]{mixed observability Markov decision process}
\acrodef{nlp}[\textsc{nlp}]{natural language processing}
\acrodef{pac}[\textsc{pac}]{probably approximately correct}
\acrodef{rl}[\textsc{rl}]{reinforcement learning}
\acrodef{ddq}[\textsc{ddq}]{dyna-Delayed Q-learning}
\acrodef{tdm}[\textsc{tdm}]{temporal Difference Models}
\acrodef{marl}[\textsc{marl}]{multi-agent reinforcement learning}
\acrodef{mmdp}[\textsc{mmdp}]{multi-agent \ac{mdp}}
\acrodef{oal}[\textsc{oal}]{optimal adaptive learning}


\maketitle
\begin{abstract}
This paper presents a theoretical framework for probably approximately correct (\textsc{pac}) multi-agent reinforcement learning (\textsc{marl}) algorithms for Markov games. 
The paper offers an extension to the well-known Nash Q-learning algorithm, using the idea of delayed Q-learning, in order to build a new \textsc{pac} \textsc{marl} algorithm for general-sum Markov games.
In addition to guiding the design of a provably \textsc{pac} \textsc{marl} algorithm, the framework enables checking whether an arbitrary \textsc{marl} algorithm is \textsc{pac}.
Comparative numerical results demonstrate performance and robustness.
\end{abstract}

\medskip


\section{Introduction} 
\label{intro}

Decision-making and planning for autonomous agents in an unknown environment is often done through a \ac{rl} approach. Sometimes systems involve more than one agent, in which case the problem fall into a \ac{marl} domain. 
In \ac{marl} one is concerned with sequential decision-making  for multiple autonomous agents that operate in an unknown environment, in which the behavior of all agents jointly affects the evolution of the system.
Such problems are typically approached from a game-theoretic perspective, with each agent trying to maximize their own reward function. 
Unlike situations where agent-environment interaction dynamics are modeled as an \ac{mdp}, in \ac{marl} games the environment of each agent is non-stationary, adding a layer of complexity to the problem. 
Indeed, existing \ac{marl} algorithms that can provide a \ac{pac} bound on the sample complexity involved in reaching a near-optimal behavior, are very rare \cite{zhang2019}. 
The paper contributes to bridging this gap with a general mathematical \ac{pac} characterization of a \ac{marl} algorithm, and a new \ac{pac} \ac{marl} algorithm for general-sum Markov games.

A \ac{marl} algorithm can either be  \cite{zhang2019} fully cooperative, fully competitive, or mixed setting. 
\emph{Fully cooperative} \ac{marl} handles cases where all agents collaborate to achieve some shared goal, sharing the same reward function--such a model is usually referred to as a \ac{mmdp}. 
As a result of shared reward function, Q-function is also identical for all agents. Hence, an early and straightforward \ac{marl} algorithm for such setting is to perform the standard Q-learning update \cite{szepesvari1999}. While Littman \cite{littman2001v} established the convergence of such algorithm to the optimal Q-values, the convergence does not necessarily imply equilibrium policy since each agent might choose a distinct equilibrium when multiple ones exist. The first \ac{marl} algorithm that provably converges to the equilibrium policy is Optimal Adaptive Learning \ac{oal} \cite{wang2003}. Recently, researchers tried to address the scalability issue that arises when the system involves big number of agents, through value function factorization \cite{sunehag2018,son2019}. Policy-based method is also developed for cooperative \ac{marl} that is provably convergent \cite{perolat2018}.

For \emph{fully competitive} \ac{marl}, most of the literature concentrated on two-agents game with zero-sum reward functions. The reason is that, there is a huge computational barrier between solving two-player zero-sum game and multi-player one. In fact, even the simplest three-player zero-sum game is known to be PPAD-complete \cite{daskalakis2009}. By defining minimax value function, \ac{marl} in two-player zero-sum games reduces into single agent case where each agent tries to maximize the worst case reward. More interestingly, optimal point of the minimax value function is the unique fixed point of a Bellman operator, which bring the strong theoretical foundation of dynamic programing in use for competitive \ac{marl}. Minimax Q-learning \cite{littman1994} extends the well-known Q-learning algorithm to zero-sum Markov games, and is provably convergent to minimax Q-values and constitutes the Nash equilibrium policy. For fully competitive setting, several \ac{pac} algorithms have also been developed. For instance, \cite{jia2019,sidford2018,sidford2019} have studied zero-sum turn-based stochastic games and with the assumption of the availability of a generative model of the game, they have proposed \ac{marl} algorithms that achieve near-optimal behavior with finite sample. Online \ac{pac} \ac{marl} algorithm has also been developed for average-reward zero-sum stochastic games using the principle of optimism in the face of uncertainty \cite{wei2017}.

In \emph{mixed setting}, finding a Nash equilibrium, as a general solution for mixed setting \ac{marl}, is actually PPAD-complete even for a simple two-player game \cite{chen2009}. 
Moreover, value-iteration based methods might fail in general to find stationary Nash or even correlated equilibrium for mixed setting \cite{zinkevich2006}. 
Existing \ac{marl} algorithms for mixed setting (e.g., Nash Q-learning \cite{hu2003}) are thus guaranteed to converge under quite strong assumptions. 
Given that finding Nash equilibria is computationally challenging \cite{savani2006,goldberg2013},
correlated Q-learning \cite{greenwald2003} circumvents Nash equilibrium computation by computing instead (via linear programming) correlated equilibria for each stage game. 
Alternative methods include Bellman residue minimization for approximation of Nash equilibria \cite{perolat2016}, and  modifications to Nash Q-learning where actions of agents are approximated by empirical averages \cite{yang2018}.

The contribution of this paper is the generalization of the \ac{pac} \ac{mdp} algorithm theorem \cite{strehl2009} to \ac{pac} Markov games. 
This serves as the basis for the introduction of a novel \ac{pac} \ac{marl} algorithm for Markov games, based on a new extension of delayed Q-learning~\cite{strehl2006} into games, which we refer to as \emph{Delayed Nash Q-learning}.
This new \ac{pac} \ac{marl} algorithm converges based on finite samples under the same conditions that Nash Q-learning \cite{hu2003} imposes.
\medskip




\section{Technical Preliminaries}  \label{pre}
A two-player finite Markov Game $M$ is a tuple $\{S,A^1,A^2,R^1,R^2,T,\gamma\}$ with elements \\
\begin{center}
\begin{threeparttable}
\begin{tabbing}
\hspace*{4.5cm} \= \kill 
$S$ \> set of \emph{states}\\
$A^i$ \> set of \emph{actions} for player $i$ \\
$R^i:S\times A^1\times A^2 \to [0,1]$ \> \emph{reward} function for player $i$\\
$T:S\times A^1\times A^2\times S \to [0,1]$ \> \emph{transition probabilities} \\
$\gamma \in [0,1)$ \> \emph{discount factor.}
\end{tabbing}
\end{threeparttable}
\end{center}
\medskip

A \emph{stationary policy} $\pi^i$ for player $i$ is a mapping $\pi^i : S\times A^i \to [0,1]$ that selects an action $a^i$ to be executed at state $s$ with its corresponding probability $\pi^{i}(s,a^i)$. 
A \emph{non-stationary policy} $\mathcal{A}^i$ can be defined as a tuple of stationary policies $\mathcal{A}^i=(\pi^i_1,\pi^i_2,\ldots)$, meaning that at step $k$ in the game, agent $i$ executes the policy $\pi^i_k$.    
In a Markov Game, each player tries to maximize its own discounted sum of rewards.
\newtheorem{definition}{Definition}

\medskip
\begin{definition}
In a two-player Markov Game $M$ where players are following policies $\pi^1$ and $\pi^2$, with action $a^i$ drawn according to policy $\pi^i$ denoted $a^i_{\sim\pi^i}$, the \emph{value of state $s$} for player $i$ is defined as: 
\begin{equation*}
v^{i}_M(s,\pi^1\!,\pi^2) \!\coloneq \mathbb{E}_{M}\left\{ \sum_{t=0}^{\infty} \gamma^{t}R^{i}\big(s_t,a^1_{\sim\pi^{1}},a^2_{\sim\pi^{2}} \big) \;\Big|\; s_0\!=\!s \right\} \enspace.
\end{equation*}
Subscript $M$ may be dropped when it clear from context.
\end{definition}

Note that by restricting the rewards in $[0,1]$, the maximum possible value of any state is bounded by $v_{\mathrm{max}}=\frac{1}{1-\gamma}$.
A Nash equilibrium is now a joint strategy in which the policy of each player is the best response to others.
\medskip

\begin{definition}
In a two-player Markov Game $M$, a \emph{Nash equilibrium point} is a tuple of policies $(\pi^{1}_{*},\pi^{2}_{*})$ such that for all state $s \in S$ and for all players $i=1,2$
\begin{equation}
\label{ne}
   \begin{dcases}
     v^{1}_M(s,\pi^1_*,\pi^2_*) \geq v^{1}_M(s,\pi^1,\pi^2_*)\\
     v^{2}_M(s,\pi^1_*,\pi^2_*) \geq v^{2}_M(s,\pi^1_*,\pi^2)
   \end{dcases} \quad  \forall \pi^i \in \Pi^i
   \enspace,
   \end{equation}
where $\Pi^i$ is the set of all available policies for player $i$. 
In addition, $(\pi^{1}_{*},\pi^{2}_{*})$ is called a \emph{global optimal} Nash equilibrium point if
\begin{equation*}
   \begin{dcases}
     v^{1}_M(s,\pi^1_*,\pi^2_*) \geq v^{1}_M(s,\pi^1,\pi^2)\\
     v^{2}_M(s,\pi^1_*,\pi^2_*) \geq v^{2}_M(s,\pi^1,\pi^2)
   \end{dcases} \quad \forall \pi^i \in \Pi^i
   \enspace.
   \end{equation*}
Moreover, $(\pi^{1}_{*},\pi^{2}_{*})$ is \emph{saddle point} Nash equilibrium, if in addition to \eqref{ne}, we have
\begin{equation*}
   \begin{dcases}
     v^{1}_M(s,\pi^1_*,\pi^2_*) \leq v^{1}_M(s,\pi^1_*,\pi^2)\\
     v^{2}_M(s,\pi^1_*,\pi^2_*) \leq v^{2}_M(s,\pi^1,\pi^2_*)
   \end{dcases} \quad \forall \pi^i \in \Pi^i \enspace.
   \end{equation*}
\end{definition}
\medskip
Every Markov Game possesses at least one Nash equilibrium point in stationary policies \cite{fink1964}. 

To adapt Q-learning into a multi-agent context, the first step is to recognize the need for considering joint actions rather than individual actions. 
For a two-agent system, for example, the Q-function is now written $Q(s,a^1,a^2)$, where the pair $(a^1,a^2)$ is referred to as the \emph{action profile}. 
With these extensions of the standard concept of Q-function, and with a Nash equilibrium as the desired solution, one can define the Nash Q-value \cite{hu2003}:
\medskip
\begin{definition}
In a two-player Markov Game $M$ with Nash equilibrium policy $(\pi^1_*,\pi^2_*)$, for any state-action profile $(s,a^1,a^2)$, the \emph{Nash Q-value} $Q^{i}_*(s,a^1,a^2)$ for agent $i$ is the expected sum of discounted rewards when both players follow the Nash equilibrium policy from next period on:
\begin{equation}
Q^{i}_*(s,a^1,a^2)= R^i(s,a^1,a^2)+ \gamma \sum_{s' \in S} T(s,a^1,a^2,s') v^{i}_M(s,\pi^1_*,\pi^2_*) \enspace.
\end{equation}
\end{definition}

Using this Definition, for all players $i$ and state $s$, one can equivalently write \cite[Lemma~10]{hu2003}
\begin{equation*}
v^i_M(s,\pi^1_*,\pi^2_*) = \sum_{(a^1,a^2)} \pi^1_*(s,a^1)\cdot\pi^2_*(s,a^2) \cdot Q^i_*(s,a^1,a^2)
\enspace.
\end{equation*}

The objective now is to design an \textsc{rl} algorithm that identifies the Nash equilibrium policy in a Markov Game when the actual transition probabilities and/or reward function are \emph{not known}. 
The procedure for finding this policy naturally involves exploration of the Markov Game model. 
An \ac{rl} algorithm usually maintains a table of state-action profile value estimates $Q^i(s,a^1,a^2)$ for all players, which are updated based on the exploration data. 
During the execution of this \ac{rl} algorithm, the currently stored value for state-action profile $(s,a^1,a^2)$ for agent $i$ at time step $t$ will be denoted $Q^i_{t}(s,a^1,a^2)$. 
Assume that all players are utilizing the same \ac{rl} algorithm and have their own estimation of Q-values for all other players.
Consequently, by considering each state of the Markov Game as a stage game (one-shot game) with $Q^i_t(s,:)$ (denoting row $s$ of the Q matrix) as the reward of each possible action profile for each player $i$, define the current value of state $s$ for player $i$ as
\begin{equation*}
v^i_{t}(s)=\mathrm{Nash}^i \big(Q^1_t(s,:),Q^2_t(s,:)\big)
\enspace,
\end{equation*}
where the $\mathrm{Nash}^i$ operator calculates the Nash Q-value for agent $i$ in the stage game with rewards $(Q^1_t(s,:),Q^2_t(s,:))$. 
We refer to a multi-agent \ac{rl} algorithm in a game-theoretic context as \emph{Nash-greedy} if, at any time step $t$ and state $s$, it instructs players to execute policies associated with some Nash equilibrium.
The policy that is in force at time step $t$ is denoted $\pi_t$, and the
greedy policy for player $i$ would be denoted in general as
\begin{equation}
\label{current policy}
\pi^i_t(s,a^i) = \mathrm{argNash}_{a^i} \big(Q^1_t(s,:),Q^2_t(s,:)\big) \enspace.
\end{equation}
\medskip


\section{Characterization of a \textsc{PAC} Multi-agent \textsc{RL} Algorithm}  \label{ddq}

This section formally frames the characteristics of the desired multi-agent \ac{pac} \ac{rl} algorithm in the form of a theorem. 
The necessary technical stage is set through the following definitions and lemmas. 

\medskip
\begin{definition}
\label{known}
Consider a two-player Markov Game $M= \{S,A^1,A^2,R^1,R^2,T,\gamma\}$, which at time step $t$ has Nash Q-value estimates $Q^i_t(s,a^1,a^2)$ for agent $i$. 
Let $K_t \subseteq S \times A^1 \times A^2$ be a set of state-action profiles which are labeled \emph{known}. 
The \emph{known state-action Markov Game} 
\begin{equation*}
M_{K_t}= \Big\{ S \cup \left\{z_{(s,a^1,a^2)} \mid (s,a^1,a^2) \notin K_t \right\}, 
\\ A^1, A^2, R^1_{K_t}, R^2_{K_t}, T_{K_t}, \gamma \Big\} 
\end{equation*}
is an Markov Game derived from $M$ and $K_t$ by defining new states $z_{s,a^1,a^2}$ for each  unknown state-action profile $(s,a^1,a^2) \notin K_t$, with self-loops for all actions, i.e.:
$
T_{K_t}(z_{(s,a^1,a^2)},\ldots,z_{(s,a^1,a^2)})= 1
$. \enspace
For all $(s,a^1,a^2)\in K_t$ and all players, it is  $R^i_{K_t}(s,a^1,a^2) = R^i(s,a^1,a^2)$ and  $T_{K_t}(s,a^1,a^2,\cdot) = T(s,a^1,a^2,\cdot)$.
When an unknown state-action profile $(s,a^1,a^2) \notin K_t$ is experienced, the reward $R^i_{K_t}(s,a^1,a^2) = Q^i_t(s,a^1,a^2)(1-\gamma)$ is accumulated for player $i$ and the model jumps to $z_{(s,a^1,a^2)}$ with $T_{K_t}(s,a^1,a^2,z_{(s,a^1,a^2)})=1$;
subsequently, $R^i_{K_t}(z_{(s,a^1,a^2)},\cdot)=Q^i_t(s,a^1,a^2)(1-\gamma)$ for all players.
\end{definition}
\medskip

Probably approximately correct (\ac{pac}) analysis of multi-agent \ac{rl} algorithms deals with the question of how fast an \ac{rl} algorithm converges to a desired fixed point.
Since in the case of this paper this point is a policy that is $\epsilon$-near or better than the Nash policy, the \ac{pac} property of a multi-agent \ac{rl} algorithm is understood here in relation to the existence of a probabilistic bound on the number of exploration steps that the algorithm takes before converging to a policy that is $\epsilon$-near or better (in terms of value) than the Nash policy.

\medskip
\begin{definition}
Consider a Markov Game $M$ with Nash equilibrium  $(\pi^1_*,\pi^2_*)$ in which all players are independently executing a given \ac{rl} algorithm $\mathcal{A}$. 
Let $s_t$ be the state visited at time step $t$ and $\mathcal{A}^i_t$ be the non-stationary policy that $\mathcal{A}$ computes for player $i$ at $t$. 
For a given $\epsilon > 0$ and $\delta >0$, $\mathcal{A}$ is a \ac{pac} if there is an $N>0$ such that with probability at least $1-\delta$, and for all but $N$ time steps, every player $i$ satisfies
\begin{equation}
v^i_M(s_t,\mathcal{A}^1_t,\mathcal{A}^2_t) \geq v^i_M(s_t,\pi^1_*,\pi^2_*) -\epsilon \enspace.
\label{epsilon-optimality}
\end{equation}
\end{definition}
\bigskip

Inequality \eqref{epsilon-optimality} is referred henceforth as the \emph{$\epsilon$-or-better Nash condition}, and $N$ as the \emph{sample complexity} of $\mathcal{A}$.
If $|\cdot|$ denotes  cardinality, the sample complexity can be a function of any combination of $|S|$ , $|A^1|$, $|A^2|$, $\frac{1}{\epsilon}$, $\frac{1}{\delta}$, and $\frac{1}{1-\gamma}$.


Now let $K_t$ be set of current known state-action profiles for an \ac{rl} algorithm $\mathcal{A}$ at time step $t$, and allow it to be arbitrarily defined as long as it depends only on the history of exploration data up to $t$. 
Any $(s,a^1,a^2) \notin K_t$ experienced at time step $t$ marks an \emph{escape event}.

The proof of each of the following three lemmas is very similar to that of the original version as it appears in literature, which the reader is referred to for details.

\newtheorem{lemma}{Lemma}
\begin{lemma} (cf. \cite[Lemma~2]{kearns2002})
\label{H-step}
For a two-player Markov Game $M$, the \emph{$H$-step value function} for policy profile $(\pi^1,\pi^2)$ is defined as
\begin{equation*}
v^{i}_M(s,\pi^1,\pi^2,H) 
\coloneq \mathbb{E}_{M}\left\{ \sum_{t=0}^{H} \gamma^{t}R^{i}\big(s_t,a^1_{\sim\pi^{1}},a^2_{\sim\pi^{2}}\big) \;\Big|\;  s_0=s \right\}
\end{equation*} 
and for $H=\frac{1}{1-\gamma}\ln{\frac{1}{(1-\gamma)\epsilon}}$ it holds
\begin{equation*}
    \big| v^{i}_M(s,\pi^1,\pi^2)-v^{i}_M(s,\pi^1,\pi^2,H) \big| \leq \epsilon  \enspace.
\end{equation*}
\end{lemma}

\begin{lemma} (cf. \cite[Lemma~9]{strehl2009})
\label{known MG bound}
Given a two-player Markov Game $M$ and a set of known state-action profiles $K_t$, the value of any state $s$ under any policy profile $(\pi^1,\pi^2)$ in the known state-action Markov Game $M_{K_t}$
is bounded by $\frac{2}{1-\gamma}=2v_{\mathrm{max}}$. 
\end{lemma}

\begin{lemma} (cf. \cite[Lemma~8]{strehl2009})
\label{coin}
Suppose that a weighted coin that is flipped has a probability $p>0$ of landing with heads up. 
Then, for any positive integer $k$ and $\delta \in (0,1)$, there exists $m=\mathcal{O} (\frac{k}{p}\ln{\frac{1}{\delta}})$, such that after $m$ tosses, with probability at least $1-\delta$, one observes $k$ or more heads. 
\end{lemma}
\medskip

 The following theorem is one of the two main results of this paper.
It offers sufficient conditions for a Nash-greedy \ac{rl} algorithm in a two-player Markov game to be \ac{pac}.
\medskip

\newtheorem{theorem}{Theorem}
\begin{theorem}
\label{pac}
Let $M$ be a two-player Markov game in which players are executing a Nash-greedy \ac{rl} algorithm $\mathcal{A}$.
Denote $(\pi^1_*,\pi^2_*)$ the Nash equilibrium of $M$, and
$K_t$ the set of current known state-action profiles at time step $t$. 
Let $M_{K_t}$ be the known state-action Markov Game at time step $t$ and  
$
\pi^i_{t} \coloneq \mathrm{argNash}_{a^i}\left(Q^1_t(s,:),Q^2_t(s,:)\right)
$
for brevity. 
Assume that $K_t = K_{t+1}$ unless at time step $t$, an update to some estimated Nash Q-value, or an escape event occurs.
Suppose also that $Q^i_t(s,a^1,a^2) \leq v_\mathrm{max}$ for all players, state-action profiles and time steps. 

If for $\epsilon>0$ the following conditions hold  with probability at least $1-\delta \in (0,1)$,
\begin{description}
    \item[\itshape optimism:] \qquad $v^i_t(s) \geq v^{i}_M(s,\pi^1_*,\pi^2_*)-\epsilon$ 
    \item[\itshape accuracy:] \qquad $v^i_t(s)-v^{i}_{M_{K_t}}(s,\pi^1_t,\pi^2_t) \leq \epsilon$ 
    \item[\itshape complexity:] \hspace{1em} the sum of the number of time steps where Nash $Q$-value updates occur plus number of time steps where escape events occur is upper bounded by $\zeta(\epsilon,\delta)>0$\enspace.
\end{description}
then the players will follow a policy that results at values which are at most $4\epsilon$ lower than a Nash policy on all but
\begin{equation}
\label{learningcom}
    \mathcal{O} \left( \frac{\zeta(\epsilon,\delta)}{\epsilon (1-\gamma)^2} \ln{(\tfrac{1}{\delta})} \ln{(\tfrac{1}{\epsilon (1-\gamma)})} \right) \simeq \mathcal{O} \left( \frac{\zeta(\epsilon,\delta)}{\epsilon (1-\gamma)^2} \right)
\end{equation}
time steps, with probability at least $1-2\delta$.
\end{theorem}

\begin{proof}
Pick $\epsilon,\delta >0$, and suppose that algorithm $\mathcal{A}$ is executed by both players in game $M$, with $\mathcal{A}^i_t$ being the current non-stationary policy of player $i$, for $i=1,2$.
Let $s_t$ denote the state of the game at $t$, 
and set $E$ be one of the two possible events that can occur in the execution of algorithm $\mathcal{A}$, $H=\frac{1}{1-\gamma}\ln{\frac{1}{(1-\gamma)\epsilon}}$ time steps after arriving at state $s_t$.
Event $E$ can be:
\begin{itemize}[leftmargin=*]
\item an update to any of the Nash Q-value estimates, or
\item an experience of a state-action profile $(s,a^1,a^2) \notin K_t$ (escape event).
\end{itemize}
Let $p_E$ be the probability that event $E$ occurs, and verify that
\begin{align*}
    v^i_M(s_t,\mathcal{A}^1_t,\mathcal{A}^2_t) &\geq v^i_M(s_t,\mathcal{A}^1_t,\mathcal{A}^2_t,H) \\
    &\geq v^i_{M_{K_t}}(s_t,\pi^1_t,\pi^2_t,H)-2v_{\mathrm{max}}\cdot p_E \enspace.
\end{align*}
The left inequality is due to all rewards being positive, and the right follows from the fact that following $\mathcal{A}^1_t$, $\mathcal{A}^2_t$ in $M$ results in an identical behavior of following $\pi^1_t$, $\pi^2_t$ in $M_{K_t}$, unless event $E$ occurs which can at most reduce the value by $2v_{\mathrm{max}}$ (Lemma~\ref{known MG bound}).
From the above, now write:
\begin{align*}
    v^i_M(s_t,\mathcal{A}^1_t,\mathcal{A}^2_t) \geq\;& v^i_{M_{K_t}}(s_t,\pi^1_t,\pi^2_t,H)-2v_{\mathrm{max}}\, p_E \\
    \stackrel{\text{Lemma}~\ref{H-step}}{\geq}& v^i_{M_{K_t}}(s_t,\pi^1_t,\pi^2_t)-\epsilon-2v_{\mathrm{max}}\, p_E \\
    \stackrel{\text{\it accuracy}}{\geq}& v^i_t(s_t)-2\epsilon-2v_{\mathrm{max}}\, p_E \\
    \stackrel{\text{\it optimism}}{\geq}& v^{i}_M(s,\pi^1_*,\pi^2_*)-3\epsilon-2v_{\mathrm{max}}\, p_E \enspace.
\end{align*}
Now, if $p_E < \frac{\epsilon}{2v_{\mathrm{max}}}$ the claim is proved: we have the $4\epsilon$-or-better Nash condition \[
v^i_M(s_t,\mathcal{A}^1_t,\mathcal{A}^2_t) \geq v^{i}_M(s,\pi^1_*,\pi^2_*)-4\epsilon\enspace.
\]
If, on the other hand, $p_E \geq \frac{\epsilon}{2v_{\mathrm{max}}}$, Lemma~\ref{coin} would guarantee with probability at least $1-\delta$ that $\zeta(\epsilon,\delta)$ occurrences of event $E$ happen within $\mathcal{O}(\frac{\zeta(\epsilon,\delta)H v_{\mathrm{max}}}{\epsilon}\ln{\frac{1}{\delta}})$ time steps. 
However, the \textit{complexity} condition guarantees with probability $1-\delta$ that $\zeta(\epsilon,\delta)$ is the maximum number of updates or escape events.
With probability at least $1-2\delta$, therefore, and for all but $\mathcal{O}(\frac{\zeta(\epsilon,\delta)H v_{\mathrm{max}}}{\epsilon}\ln{\frac{1}{\delta}}) =  \mathcal{O}(\frac{\zeta(\epsilon,\delta)}{\epsilon (1-\gamma)^2} \ln{(\tfrac{1}{\delta})} \ln{(\tfrac{1}{\epsilon (1-\gamma)})})$ timesteps, it must be the case that
$v^i_M(s_t,\mathcal{A}^1_t,\mathcal{A}^2_t) \geq v^{i}_M(s,\pi^1_*,\pi^2_*)-4\epsilon$. 
\end{proof}
\medskip


\section{Delayed Nash Q-learning Algorithm}  \label{dnq}
This section presents the second key contribution of this paper: a new \ac{rl} algorithm for two-player Markov Games, referred to as Delayed Nash Q-learning (Algorithm~\ref{alg}). 
Delayed Nash Q-learning is the first algorithm that is \ac{pac} in terms of converging to a policy which is arbitrarily near (or better than) the Nash equilibrium policy. 
The computational complexity of Delayed Nash Q-learning is roughly the same as the well-known Nash Q-learning \cite{hu2003}.

We assume that both players are executing the (same) algorithm. 
Thus, players keep the Nash Q-value estimates of their opponents. 
This is made possible by the assumption that every player can observe the rewards of all players at each step in the Markov Game. 
With the same observations, players' estimates of the Nash Q-values are identical.

As the ``delayed'' term in the name suggests, the algorithm waits until a state-action profile is experienced $m$ times before it makes an update of its Nash Q-value.
The update mechanism (of lines $24,25$) requires a successful update to  change the Nash Q-value estimate, triggered by another parameter $\epsilon_1$. 
Both $m$ and $\epsilon_1$ parameters are tunable.
Similarly to Delayed Q-learning \cite{strehl2009} or R-max \cite{brafman2002}, the algorithm reported here utilizes the principle of ``optimism in the face of uncertainty'' to encourage exploration by originally over-estimating the Nash Q-value estimates to some $v_{\mathrm{max}}$.

Similar to Delayed Q-learning \cite{strehl2009}, Delayed Nash Q-learning maintains the following internal variables: 
\begin{itemize}[leftmargin=0.15in]
    \item $l^i(s,a^1,a^2)$ is the number of samples gathered for the update of $Q^i(s,a^1,a^1)$ once it is the case that $l^i(s,a^1,a^2)=m$.
    \item $U^i(s,a^1,a^2)$ stores the running sum of target values that will be used for the update of $Q^i(s,a^1,a^1)$ once enough samples have been gathered.
    \item $b^i(s,a^1,a^2)$ is the time step at which the collection of the most recent $m$ experiences of $(s,a^1,a^2)$ started.
    \item $\mathsf{learn}^i(s,a^1,a^2)$ is a Boolean flag  indicating whether samples are being gathered for state-action profile $(s,a^1,a^2)$. 
    It is set to $\mathrm{true}$ initially, and is reset to $\mathrm{true}$ whenever some Nash Q-value of player $i$ is updated.
    It changes to $\mathrm{false}$ when no updates to any Nash Q-values occur within a time window in which $(s,a^1,a^2)$ is experienced $m$ times, but subsequent attempted updates of $Q^i(s,a^1,a^2)$ fail (cf.~\cite{strehl2009}).
\end{itemize}

  
\begin{algorithm}
\caption{The Delayed Nash Q-learning algorithm \label{alg}}
\begin{algorithmic}[1]
{\normalsize{
\State \textbf{Inputs}: $S,A^1,A^2,\gamma,m,\epsilon_1$
\For{\textbf{all} $s,a^1,a^2,i=1,2$}
\State $Q^i(s,a^1,a^2) \gets v_\mathrm{max}$ \Comment{{\footnotesize{set Nash $Q$-value to its maximum}}}
\State $U^i(s,a^1,a^2) \gets 0$ \Comment{{\footnotesize{used for attempted updates}}}
\State $l^i(s,a^1,a^2) \gets 0$ \Comment{{\footnotesize{counters}}}
\State $b^i(s,a^1,a^2) \gets 0$ \Comment{{\footnotesize{beginning timestep of attempted update}}}
\State $\mathsf{learn}^i(s,a^1,a^2) \gets \mathrm{true}$ \Comment{{\footnotesize{learn flags}}}
\EndFor
\State $t^* \gets 0$ \Comment{{\footnotesize{time of the most recent successful update}}}
\For{$t=1,2,3,...$}
\State let $s$ denotes the state at time $t$
\State players choose $(a^1,a^2)$ according to $\pi^i_t$ as in \eqref{current policy}
\State observe immediate rewards $r^1,r^2$ and next state $s'$
\If{$b^i(s,a^1,a^2) \leq t^*$}
\State $\mathsf{learn}^i(s,a^1,a^2) \gets \mathrm{true}$
\EndIf
\If{$\mathsf{learn}^i(s,a^1,a^2)=\mathrm{true}$}
\If{$l^i(s,a^1,a^2)=0$}
\State $b^i(s,a^1,a^2) \gets t$
\EndIf
\State $l^i(s,a^1,a^2) \gets l^i(s,a^1,a^2)+1$
\State $U^i(s,a^1,a^2) \gets U^i(s,a^1,a^2)+r^i+\gamma v^i(s')$
\If{$l^i(s,a^1,a^2)=m$}
\If{$Q^i(s,a^1,a^2)-U^i(s,a^1,a^2)/m \geq 2 \epsilon_1$}
\State $Q^i(s,a^1,a^2) \gets U^i(s,a^1,a^2)/m + \epsilon_1$ \ 
\State $t^* \gets t$
\ElsIf{$b^i(s,a^1,a^2) > t^*$}
\State $\mathsf{learn}^i(s,a^1,a^2) \gets \mathrm{false}$
\EndIf
\State $U^i(s,a^1,a^2) \gets 0$ 
\State $l^i(s,a^1,a^2) \gets 0$
\EndIf
\EndIf
\EndFor
}}
\end{algorithmic}
\end{algorithm}
\medskip


\section{PAC Properties of Delayed Nash Q-learning}  \label{dnq:pac}

In this section, we claim that Delayed Nash Q-learning algorithm is \ac{pac} and we present its sample complexity bound.
The proof of this result requires the following assumption, which is the same requirement that the well-known Nash Q-learning needs to guarantee convergence \cite{hu2003}.
\newtheorem{assumption}{Assumption}

\begin{assumption}
\label{gs}
For state $s$ and time step $t$, Delayed Nash Q-learning encounters a stage game associated with $\big(Q^1_t(s,:),Q^2_t(s,:)\big)$, which possesses either a global optimum or a saddle point Nash equilibrium. 
The algorithm may use either of the two for its updates. 
\end{assumption}

Before formally stating the \ac{pac} properties of the Delayed Nash Q-learning algorithm and proving the bound on its sample complexity, some technical groundwork needs to be laid. 
To slightly simplify notation, let 
\begin{align*}
\kappa &\triangleq \frac{|S||A^1||A^2|}{(1-\gamma)\epsilon_1}
&
v^i_*(s) &\triangleq v^i_M(s,\pi^1_*,\pi^2_*) \enspace.
\end{align*}
and note that subscript $t$ marks the value of a variable at the \emph{beginning} of time step $t$ (particularly line $17$ of the algorithm) .

\begin{definition}
An \emph{attempted update} is an event at which $\mathsf{learn}^i(s,a^1,a^2)=\mathrm{true}$ and $l^1(s,a^1,a^2)=l^2(s,a^1,a^2)=m$. An attempted update can be successful or unsuccessful depending on the condition of line $24$ of the algorithm. 
\end{definition}

\begin{definition}
\label{definition:beta}
At any time step $t$ of the Delayed Nash Q-learning algorithm, the set of \emph{known state-action profiles} is defined as 
\begin{multline}
\label{kt}
    K_t \coloneq K^1_t \cap K^2_t \quad \text{where} \quad 
    K^i_t \coloneq \Big\{ (s,a^1,a^2) \mid Q^i_t(s,a^1,a^2) - R^i(s,a^1,a^2) -\gamma \textstyle{\sum_{s'}T(s,a^1,a^2,s')v^i_t(s')} \leq 3 \epsilon_1 \Big\}
\end{multline}
\end{definition}

For the Delayed Nash Q-learning algorithm a number of facts can be shown. 
The proof for these claims can be found in different appendices at the end of this paper.
First, the number of successful updates is bounded:
\begin{lemma}
\label{sbound}
The total number of updates during the execution of Delayed Nash Q-learning algorithm is bounded by $2\kappa$.
\end{lemma}
\begin{proof}
In Appendix~\ref{L1}.
\end{proof}

Attempted updates are bounded in number:
\begin{lemma}
\label{ubound} 
The total number of attempted updates in Delayed Nash Q-learning algorithm is bounded by $2|S||A^1||A^2|(1+2\kappa)$.
\end{lemma}
\begin{proof}
In Appendix~\ref{L2}.
\end{proof}

Players decrease their state Nash value estimates as time goes on:
\begin{lemma}
\label{decreas} 
Let $t_1 < t_2$ be two timesteps during the execution of Delayed Nash Q-learning algorithm. Then under the Assumption~\ref{gs}, for all states $s$ and any player $i$,
$
    v^i_{t_1}(s) \geq v^i_{t_2}(s)
$.
\end{lemma}
\begin{proof}
In Appendix~\ref{L3}.
\end{proof}

\begin{lemma}
\label{optimism} 
Suppose that the Delayed Nash Q-learning is executed on a Markov Game $M$ under Assumption~\ref{gs} with parameter $m$ satisfying
\begin{equation}
\label{m}
    m \geq \frac{\ln{\big(\tfrac{6|S||A^1||A^2|(1+2\kappa)}{\delta}\big)}}{2 \epsilon^2_1 (1-\gamma)^2} \simeq \mathcal{O} \left( \frac{\ln{\big(\tfrac{|S|^2|A^1|^2|A^2|^2}{\delta}\big)}}{\epsilon^2_1 (1-\gamma)^2} \right)
    \enspace.
\end{equation}
Then, $Q^i_t(s,a^1,a^2) \geq Q^i_*(s,a^1,a^2)$ and $v^i_t(s) \geq v^i_*(s)$ for any player $i$, state-action profile $(s,a^1,a^2)$, and timestep $t$, with probability at least $1-\frac{\delta}{3}$.
\end{lemma}
\begin{proof}
In Appendix~\ref{L4}.
\end{proof}

\begin{lemma}
\label{suc}
Under Assumption~\ref{gs} and with the choice of $m$ as in \eqref{m}, assume that the Delayed Nash Q-learning algorithm is at timestep $t$ with $(s,a^1,a^2) \notin K^i_t$, $l^i(s,a)=0$ and $\mathsf{learn}^i(s,a^1,a^2)=\mathrm{true}$ for player $i$. 
Knowing that an attempted update of $Q^i(s,a^1,a^2)$ will necessarily occur within $m$ occurrences of $(s,a^1,a^2)$ after $t$, say at timestep $t_{m}$, this attempted update at $t_{m}$ will be successful with probability at least $1-\frac{\delta}{3}$. 
\end{lemma}
\begin{proof}
In Appendix~\ref{L5}.
\end{proof}

\begin{lemma}
\label{escape1}  
Let $t$ be the timestep when an unsuccessful update of $Q^i(s,a^1,a^2)$ occurs after the conditions of Lemma~\ref{suc} were satisfied. 
If $\mathsf{learn}^i(s,a^1,a^2)=\mathrm{false}$ at timestep $t+1$, then $(s,a^1,a^2) \in K^i_{t+1}$. 
\end{lemma}
\begin{proof}
In Appendix~\ref{L6}.
\end{proof}

\begin{lemma}
\label{ebound}
During the execution of the Delayed Nash Q-learning algorithm, and assuming that Lemma~\ref{suc} applies, the total number of timesteps with $(s_t,a^1_t,a^2_t) \notin K_t$ (i.e. escape events) is at most $4m\kappa$.
\end{lemma}
\begin{proof}
In Appendix~\ref{L7}.
\end{proof}


\color{black}

The \ac{pac} properties of Algorithm~\ref{alg} can now be established in following form theorem, the proof of which is based on an of Theorem~\ref{pac}.

\begin{theorem}
\label{nash pac}
Consider a two-player Markov Game $M= \{S,A^1,A^2,T,R^1,R^2,\gamma\}$.
Pick $\epsilon \in \left(0,\tfrac{1}{1-\gamma}\right)$, and $\delta \in (0,1)$. 
Then, with $\frac{1}{\epsilon_1}=\tfrac{3}{(1-\gamma)\epsilon}=\mathcal{O}\left(\sfrac{1}{\epsilon(1-\gamma)}\right)$, there exists an integer
\[ 
m= \mathcal{O}\left(\tfrac{\ln{(\sfrac{|S|^2|A^1|^2|A^2|^2}{\delta})}}{\epsilon^2_1(1-\gamma)^2}\right)
\enspace,
\]
such that if the Delayed Nash Q-learning algorithm is executed by both players under Assumption~\ref{gs} and with the set $K_t$ defined as \eqref{kt}, the players will follow a policy which with probability at least $1-2 \delta$ is at most $4 \epsilon$ worse than a Nash policy, on all but 
\begin{equation}
\label{nashcom}
    \mathcal{O} \left(\tfrac{|S||A^1||A^2|}{\epsilon^4(1-\gamma)^8} \right)
\end{equation}
time steps (logarithmic factors ignored).
\end{theorem}

\begin{proof}
Apply Theorem~\ref{pac}.
It is already shown in Lemma~\ref{optimism} that the \emph{optimism} condition holds throughout the execution of Delayed Nash Q-learning algorithm.

To establish the \emph{accuracy} condition, for all player $i=1,2$ and all state $s$, write
\begin{multline}
\label{vmkt}
    v^i_{M_{K_t}}(s,\pi^1_t,\pi^2_t) = \hspace{-1.2em} \sum_{\tiny{\substack{(a^1,a^2)\\(s,a^1,a^2)\in K_t}}}
    \hspace{-1em}
    \pi^1_{t}(s,a^1)\,\pi^2_{t}(s,a^2) 
    \left[ R^i(s,a^1,a^2)+\gamma \sum_{s'}T(s,a^1,a^2,s')v^i_{M_{K_t}}(s',\pi^1_t,\pi^2_t) \right] \\
    + \sum_{\tiny{\substack{(a^1,a^2)\\(s,a^1,a^2)\notin K_t}}}
    \hspace{-1em}
    \pi^1_{t}(s,a^1)\,\pi^2_{t}(s,a^2)\,Q^i_{t}(s,a^1,a^2)
\end{multline}
Similarly,
\begin{align}
\label{vt}
    v^i_t(s) 
    &= 
    \sum_{\tiny{(a^1,a^2)}} \pi^1_{t}(s,a^1)\,\pi^2_{t}(s,a^2) \, Q^i_{t}(s,a^1,a^2) \notag \\
    &=
    \hspace{-1.2em}
    \sum_{\tiny{\substack{(a^1,a^2)\\(s,a^1,a^2)\in K_t}}}
    \hspace{-1em}
    \pi^1_{t}(s,a^1) \pi^2_{t}(s,a^2) 
    \left[ R^i(s,a^1,a^2)+\gamma \sum_{s'}T(s,a^1,a^2,s')v^i_t(s') +\beta^i(s,a^1,a^2) \right] \notag\\
    & \qquad 
    + \sum_{\tiny{\substack{(a^1,a^2)\\(s,a^1,a^2)\notin K_t}}}
    \hspace{-1em}
    \pi^1_{t}(s,a^1)\,\pi^2_{t}(s,a^2)\, Q^i_{t}(s,a^1,a^2)
    \enspace,
\end{align}
 where by the definition of $K_t$ as \eqref{kt}, you know that $\beta^i(s,a^1,a^2) \leq 3\epsilon_1$. To slightly simplify the notation, set $\Delta^i(s) \triangleq v^i_t(s)-v^i_{M_{K_t}}(s,\pi^1_t,\pi^2_t)$. Now, if you denote
 \begin{equation*}
     \alpha \coloneq \max_{s} \big(\Delta^i(s)\big) = \Delta^i(s^*)
     \enspace,
 \end{equation*}
by \eqref{vmkt} and \eqref{vt}, it follows that $\alpha$ affords the bound
 \begin{align*}
\label{alpha}
    \alpha &= \sum_{\tiny{\substack{(a^1,a^2)\\(s^*,a^1,a^2)\in K_t}}} \pi^1_{t}(s^*,a^1)\, \pi^2_{t}(s^*,a^2) 
    \left[
    \gamma \sum_{s'}T(s^*,a^1,a^2,s')\Delta^i(s') + \beta^i(s^*,a^1,a^2) \right] \\
     &\leq \gamma \alpha +3\epsilon_1
     \enspace,
\end{align*}
from which it follows that
$
\alpha \le \gamma \alpha + 3 \epsilon_1 \implies 
\alpha \le \tfrac{ 3\epsilon }{ 1 - \gamma } = \epsilon
$.

Finally, to confirm the \emph{complexity} condition, invoke Lemmas~\ref{sbound} and \ref{ebound} to see that the learning complexity $\zeta (\epsilon,\delta)$ is bounded by $2\kappa + 4m\kappa$. 

In conclusion, the conditions of Theorem \ref{pac} are satisfied and therefore the Delayed Nash Q-learning algorithm is \ac{pac}. 
Substituting $\zeta (\epsilon,\delta)$ into \eqref{learningcom} yields \eqref{nashcom} and completes the proof.
\end{proof}
\medskip


\section{Simulation results}  \label{sim}

In this section, the performance of the Delayed Nash Q-learning algorithm is evaluated on two grid-world games used in work reported in literature \cite{hu2003}.

The grid-world games 1 and 2 are shown in Figs.~\ref{gg1} and \ref{gg2}, respectively, where the initial and goal positions of each player is depicted. 
In grid-world game $1$ the players have different goal positions; in $2$, they have the same. 
Cells are labeled in the order depicted in Fig.~\ref{gg-order}.
Both players can move only one cell a time, exercising one of four primitive actions: {\small \textsf{down}} ($\mathsf{d}$), {\small \textsf{left}} ($\mathsf{l}$), {\small \textsf{up}} ($\mathsf{u}$), and {\small \textsf{right}} ($\mathsf{r}$). 
If the two players attempt to move together into any cell other than the goal, they bounce back to their previous cells.
\begin{figure}[h!]
    \centering
	\begin{subfigure}[b]{0.25\textwidth}
    \centering
    \includegraphics[width=1\textwidth]{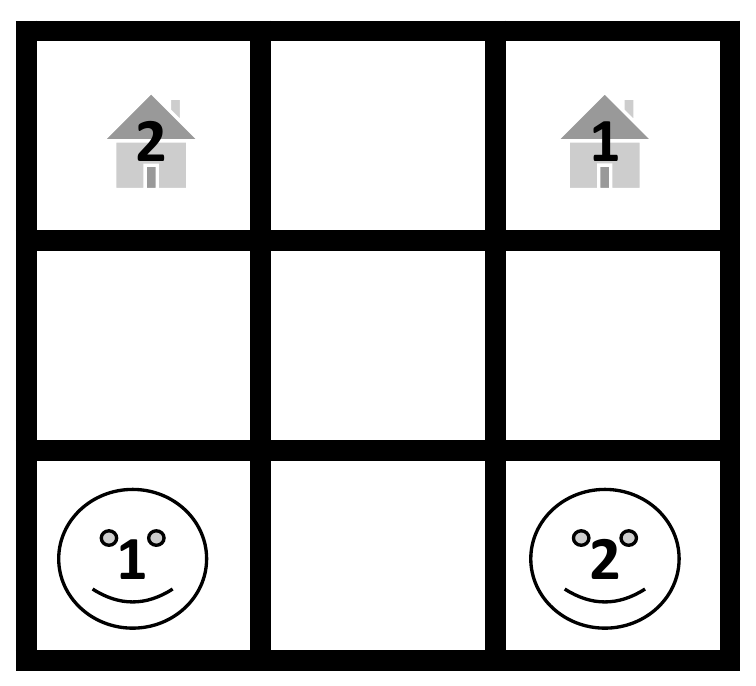}
	\caption{ Grid world 1
 	\label{gg1}
	}
    \end{subfigure} 
	\begin{subfigure}[b]{0.25\textwidth}
    \centering
    \includegraphics[width=1\textwidth]{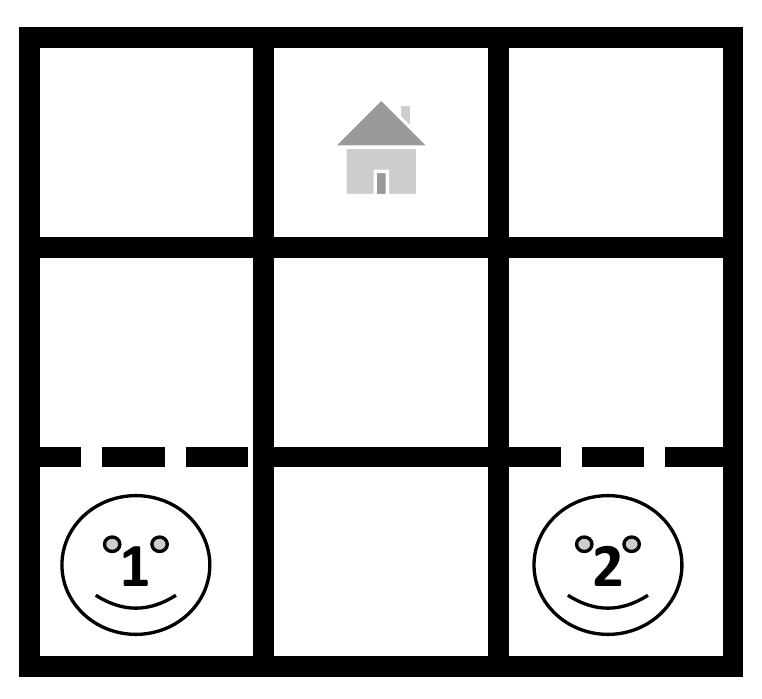}
	\caption{ Grid world 2
 	\label{gg2}
	}	
    \end{subfigure}
	\begin{subfigure}[b]{0.25\textwidth}
    \centering
    \includegraphics[width=1\textwidth]{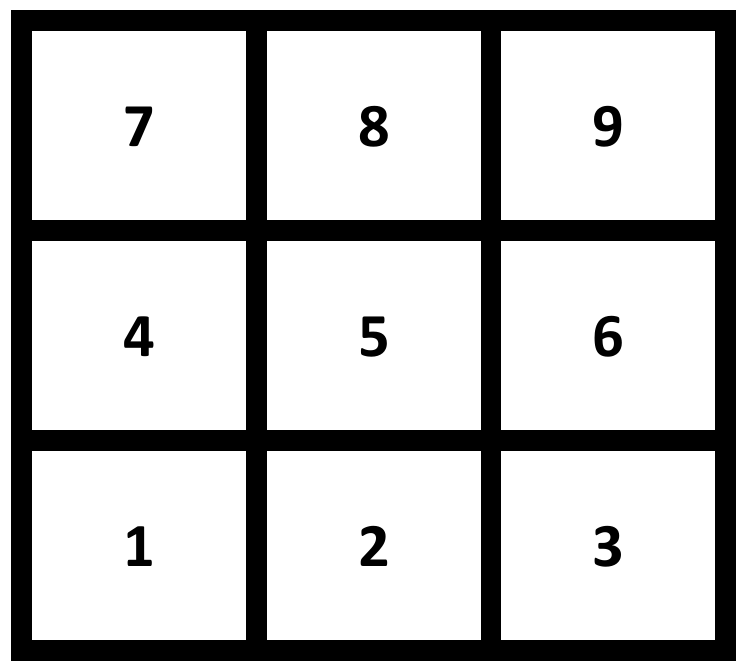}
	\caption{ Cell indexing
 	\label{gg-order}
	}	
    \end{subfigure}    
\caption{The two grid-world games where Delayed Nash Q-learning is tested (cf.~\cite{hu2003}).
The home cells are the goal configurations for the two agents and the smiley faces mark their initial positions.}
\end{figure}

The state-space of the game is the set of pairs $S=\{(1,2),(1,3),\ldots,(9,8)\}$, in which the first number is the location of player $1$ and the second number is the location of player $2$. 
State $(1,3)$ is the initial state, and all states in which one of the players is in its goal location are terminal states. 
If a player reaches its goal position, it receives the reward of $1$, and for all other moves the reward is $0$. 
Transitions in grid-world game $1$ are deterministic; transitions in grid-game $2$ are also deterministic except for action {\small \textsf{up}} in cells $1$ and $3$ (shown by dashed lines in Fig.~\ref{gg2}).
If player 1 (resp.~2) chooses action {\small \textsf{up}} in cell $1$ (resp.~3), it will either move to cell 4 (resp.~6) with probability $0.5$ or stay in the same cell with probability $0.5$.       

Multiple Nash equilibrium strategies exist for grid-world game $1$. 
Some examples  are depicted in Fig.~\ref{nash1}. 
For grid-world game $2$ on the other hand, there exist only two Nash equilibrium strategies: the ones shown in Fig.~\ref{nash2}.
\begin{figure}[h!]
    \centering
    \includegraphics[width=0.25\textwidth]{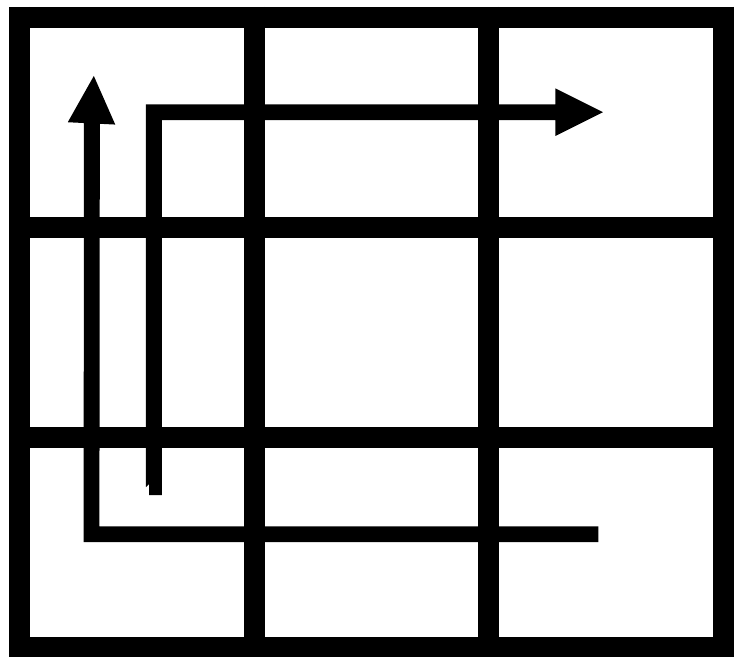}
    \includegraphics[width=0.25\textwidth]{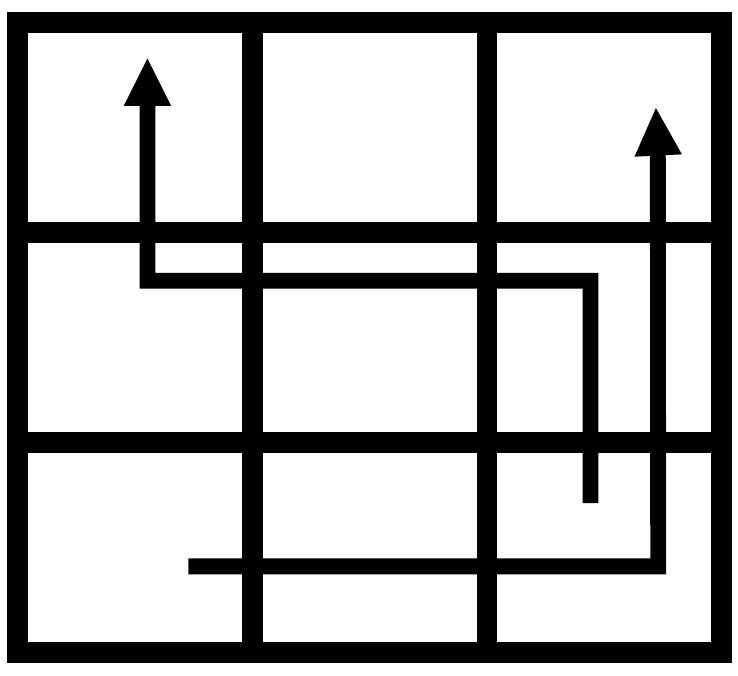}
    \caption{Examples of Nash strategies for grid-world  $1$.\label{nash1} }
\end{figure}
\begin{figure}[h!]
    \centering
    \includegraphics[width=0.25\textwidth]{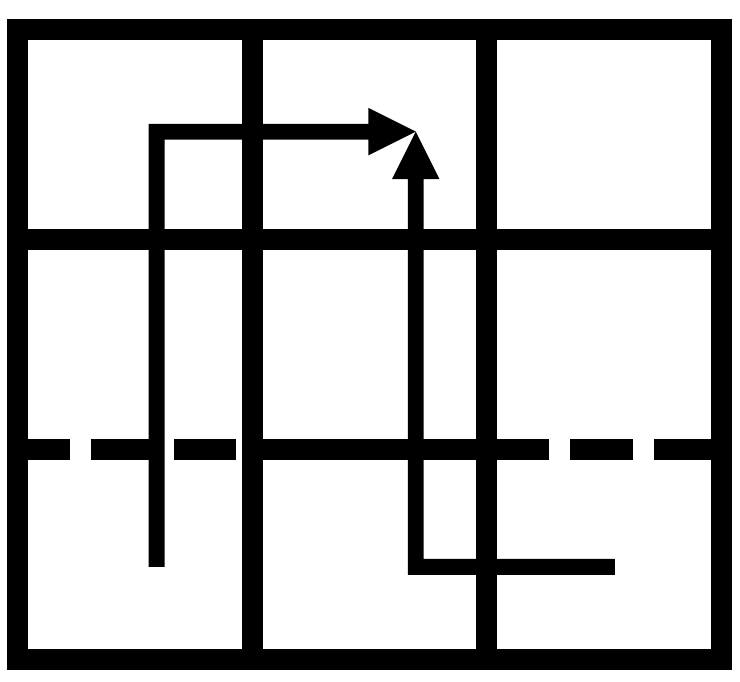}
    \includegraphics[width=0.25\textwidth]{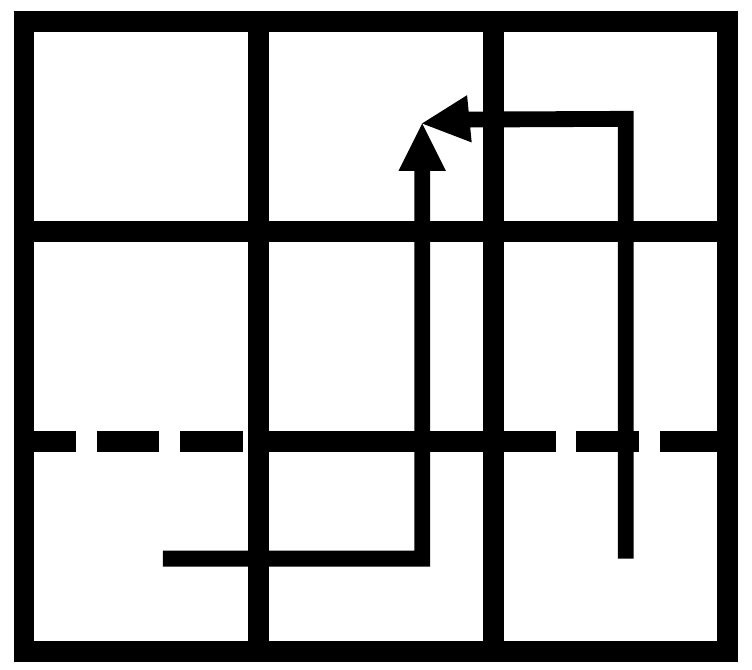}
    \caption{The Nash equilibrium strategies for grid-world  $2$. \label{nash2}}
\end{figure}

Note that all Nash equilibria in grid-world game $1$ are globally optimal; however, there is no guarantee that all stage games encountered during the execution of the learning algorithm possess a globally optimal Nash equilibrium \cite{hu2003}. 
On the other hand, none of the Nash equilibrium strategies for grid-world game $2$ are globally optimal or saddles.

In the earlier study of these two grid-world games, it has been reported \cite{hu2003} that while both games do not necessarily satisfy Assumption~\ref{gs}, the convergence of a Nash Q-learning algorithm is almost guaranteed for the game that has an overall globally optimal or saddle Nash equilibrium strategy (e.g., grid-world game $1$). 
For games that do not, (e.g., grid-world game $2$), a $79\%$ probability of convergence using a Nash Q-learning algorithm was reported \cite{hu2003}.  

In this paper, Delayed Nash Q-learning algorithm was implemented on the two grid-world games with the following parameters: $\gamma=0.8$, $m=50$ and $\epsilon=0.06$. 
The algorithm was run on each game for $50$ times (matching \cite{hu2003}). 
Delayed Nash Q-learning converged to a Nash equilibrium strategy profile, in both games, \emph{every time}. 
The average number of steps needed for convergence (the empirical sample complexity of the algorithm) was 445640 in the case of grid-world game 1, and 485460 in the case of grid-world game 2.
\medskip


\section{Conclusion}  \label{section:conclusion}

Computational gains and analytical performance guarantees obtained recently in applications of \ac{rl} to finite \ac{mdp}s can be extended to Markov games yielding a unique new \ac{pac} algorithm for learning game equilibria.
A new sample-efficient \ac{marl} algorithm for Markov games emerges as the  outcome of careful integration of delayed Q-learning techniques with Nash Q-learning methodologies.
This algorithm, referred to as delayed Nash Q-learning, is guaranteed to converge on finite samples, under the same assumptions utilized for Nash Q-learning.
Numerical results not only support the theoretical \ac{pac} predictions, but also indicate that the delayed Nash Q-learning algorithm performance may degrade much more gracefully than that of Nash Q-learning when the underlying assumptions are violated.
\medskip



\bibliographystyle{unsrt}
\bibliography{My_Collection}






\appendix
\appendixpage

\section{Proof of Lemma~\ref{sbound}} 
\label{L1}
Consider a fixed state-action profile $(s,a^1,a^2)$. 
Its value $Q^i(s,a^1,a^2)$ for player $i$ is initially set to $v_\mathrm{max}=\frac{1}{1-\gamma}$.
When an update is successful $Q^i(s,a^1,a^2)$ is reduced by at least $\epsilon_1$ (because of the condition of line $24$ of the algorithm). 
Since the reward function $R^i(s,a)$ is non-negative,  $Q^i(s,a^1,a^2) \geq 0$ in all timesteps, which means that there can be at most $\lfloor \frac{1}{\epsilon_1(1-\gamma)} \rfloor$ updates for $Q^i(s,a^1,a^2)$.
With $|S| |A^1| |A^2|$ total state-action profiles and since there are $2$ players, the total number updates is bounded by $2\kappa = 2\frac{|S||A^1||A^2|}{(1-\gamma)\epsilon_1}$. 

\section{Proof of Lemma~\ref{ubound}} 
\label{L2} 
Suppose an attempted update  occurs at timestep $t$ to some $Q^i(s,a^1,a^2)$. 
By definition, in order for a subsequent attempted update to $Q^i(s,a^1,a^2)$ to occur at timestep $t' > t$, at least one update to any Nash Q-value estimate must occur between $t$ and $t'$.
Lemma \ref{sbound} ensures that there can be no more than $2\kappa$ (successful) updates.
In other words, the most frequent occurrence of an attempted update is interlaced between successful updates, which implies that at most $1+2\kappa$ attempted updates are possible for $Q^i(s,a^1,a^2)$. 
Scaling this argument to all  state-action profiles and both players, we arrive at the $2|S||A^1||A^2|(1+2\kappa)$ upper bound.

\section{Proof of Lemma~\ref{decreas}} 
\label{L3}
Show the lemma for player $1$; the proof for the other player is identical.
First note that because the Delayed Nash Q-learning algorithm only allows updates that decrease the $Q$ estimates, for all state $s \in S$
\begin{equation*}
    Q^1_{t_1}(s,a^1,a^2) \geq Q^1_{t_2}(s,a^1,a^2) \quad \forall (a^1,a^2) \in A^1 \times A^2
    \enspace.
\end{equation*}
Let $(\pi^1_{t_1},\pi^2_{t_1})$ and $(\pi^1_{t_2},\pi^2_{t_2})$ be the Nash equilibria of stage games associated with state $s$, at timesteps $t_1$ and $t_2$, respectively. 
Then,  write
\begin{align*}
    v^1_{t_1}(s)= \sum_{(a^1,a^2)} \pi^1_{t_1}(s,a^1)\cdot\pi^2_{t_1}(s,a^2) \cdot Q^1_{t_1}(s,a^1,a^2) \\
    v^1_{t_2}(s)= \sum_{(a^1,a^2)} \pi^1_{t_2}(s,a^1)\cdot\pi^2_{t_2}(s,a^2) \cdot Q^1_{t_2}(s,a^1,a^2) 
    \enspace.
\end{align*}
If the stage games possess global optimal Nash equilibrium, it will be
\begin{align*}
    v^1_{t_1}(s) &\stackrel{\text{\tiny{global optimal}}}{\geq} \sum_{(a^1,a^2)} \pi^1_{t_2}(s,a^1)\,\pi^2_{t_2}(s,a^2) \,Q^1_{t_1}(s,a^1,a^2)
    &\\ 
    &\geq  \sum_{(a^1,a^2)} \pi^1_{t_2}(s,a^1)\,\pi^2_{t_2}(s,a^2) \,Q^1_{t_2}(s,a^1,a^2) & = v^1_{t_2}(s)
    \enspace,
\end{align*}
whereas if the stage games possess saddle point Nash Equilibrium, $v^1_{t_1}$ is bounded as
\begin{align*}
    v^1_{t_1}(s) &\stackrel{\text{\tiny{Nash equilibrium}}}{\geq} \sum_{(a^1,a^2)} \pi^1_{t_2}(s,a^1)\,\pi^2_{t_1}(s,a^2) \,Q^1_{t_1}(s,a^1,a^2) &\\ 
    &\geq  \sum_{(a^1,a^2)} \pi^1_{t_2}(s,a^1)\,\pi^2_{t_1}(s,a^2) \,Q^1_{t_2}(s,a^1,a^2) &\\
    &\stackrel{\text{\tiny{saddle}}}{\geq}  \sum_{(a^1,a^2)} \pi^1_{t_2}(s,a^1)\,\pi^2_{t_2}(s,a^2) \,Q^1_{t_2}(s,a^1,a^2) &= v^1_{t_2}(s)
    \enspace.
\end{align*}
Thus in any case, it will be $v^1_{t_1}(s) \ge v^1_{t_2}(s)$ and the proof is completed.

\section{Proof of Lemma~\ref{optimism}} 
\label{L4}
Prove the claim for player $1$; the proof for the other player is identical. 
The proof involves strong induction for all state-action profile $(s,a^1,a^2)$:
\begin{compactenum}[(i)]
\item At $t=1$, the values of all state-action profiles are set to the maximum possible Nash value of the Markov Game $M$. 
This implies that $Q^1_1(s,a^1,a^2) \geq Q^1_*(s,a^1,a^2)$ and $v^1_1(s) \geq v^1_*(s)$.
\item Assume that $Q^1_t(s,a^1,a^2) \geq Q^1_*(s,a^1,a^2)$ and $v^1_t(s) \geq v^1_*(s)$ for all timesteps up to and including $t=n-1$.
\item If no successful update happens during the timestep $t=n-1$, then 
\[
Q^1_n(s,a^1,a^2) = Q^1_{n-1}(s,a^1,a^2) \geq Q^1_*(s,a^1,a^2)\;, \qquad v^1_n(s)=v^1_{n-1} \geq v^1_*(s)
\]
and the claim is immediately established.
If not, assume that during the timestep $t=n-1$, the state-action profile for which the update occurs is $(s,a^1,a^2)$. 
Suppose that the latest $m$ experiences of $(s,a^1,a^2)$ happened at timesteps $t_1 <t_2< \cdots <t_{m}=n-1$, at which the player was rewarded 
$r^1[1], r^1[2], \ldots, r^1[m]$ and the game jumped to states $s[1], s[2], \ldots, s[m]$, respectively. 
Define the random variable $Y\coloneq r^1[i]+\gamma v^1_*(s[i])$ for $1 \leq i \leq m$ and note that $0 \leq Y \leq \frac{1}{1-\gamma}$. 
Then a direct application of the Hoeffding inequality for bounded random variables and with the choice of $m$ as in \eqref{m} implies
\begin{equation*}
    \frac{1}{m} \sum^{m}_{i=1} \big( r^1[i]+\gamma v^1_*(s[i]) \big) >
    \\ \mathbb{E} \big\{Y\big\} - \epsilon_1 = Q^1_*(s,a^1,a^2)-\epsilon_1
\end{equation*}
with probability $1-\sfrac{\delta}{6\big(|S||A^1||A^2|(1+2\kappa)\big)}$.
Now you have:
\begin{align*}
    Q^1_n(s,a^1,a^2)& =\frac{1}{m} \big( \sum^{m}_{i=1} r^1[i]+\gamma v^1_{t_i}(s[i]) \big) + \epsilon_1 
    \\ &\geq \frac{1}{m} \big( \sum^{m}_{i=1} r^1[i]+\gamma v^1_*(s[i]) \big)+\epsilon_1 
    \\ &\geq Q^1_*(s,a^1,a^2) - \epsilon_1 + \epsilon_1 \\ &= Q^1_*(s,a^1,a^2)
    \enspace.
\end{align*}
To show $v^1_n(s)\geq v^1_*(s)$, first note that
\begin{equation*}
    v^1_n(s)= \sum_{(a^1,a^2)} \pi^1_n(s,a^1) \pi^2_n(s,a^2) Q^1_n(s,a^1,a^2)
    \enspace,
\end{equation*}
where $(\pi^1_n,\pi^2_n)$ denotes the policy executed at timestep $t=n$. 
If the stage games possess global optimum Nash equilibrium, 
\begin{align*}
    v^1_n(s) &\stackrel{\text{\tiny{global optimum}}}{\geq} \sum_{(a^1,a^2)}\ \pi^1_*(s,a^1)\, \pi^2_*(s,a^2)\, Q^1_n(s,a^1,a^2) & \\ 
    &\geq  \sum_{(a^1,a^2)} \pi^1_*(s,a^1) \,\pi^2_*(s,a^2) \,Q^1_*(s,a^1,a^2) 
    &= v^1_*(s)
    \enspace.
\end{align*}
If the stage games possess a saddle point Nash Equilibrium, 
\begin{align*}
    v^1_n(s) &\stackrel{\text{\tiny{Nash equilibrium}}}{\geq} \sum_{(a^1,a^2)} \pi^1_*(s,a^1)\,\pi^2_n(s,a^2) \,Q^1_n(s,a^1,a^2) & \\ 
    &\geq  \sum_{(a^1,a^2)} \pi^1_*(s,a^1)\,\pi^2_n(s,a^2) \,Q^1_*(s,a^1,a^2) & \\
    &\stackrel{\text{\tiny{saddle}}}{\geq}  \sum_{(a^1,a^2)} \pi^1_*(s,a^1)\,\pi^2_*(s,a^2) \,Q^1_*(s,a^1,a^2)
    &= v^1_*(s) \enspace.
\end{align*} 
\end{compactenum}
The induction argument is thus complete.
Since the conclusion has to be true for all possible attempted updates,  invoke Lemma~\ref{ubound}, according to which $2|S||A^1||A^2|(1+2\kappa)$ is an upper bound for all possible attempted updates. 
Therefore, the statement above is true with probability at least  $\big(1-\sfrac{\delta}{6\big(|S||A^1||A^2|(1+2\kappa)\big)}\big)^{2|S||A^1||A^2|(1+2\kappa)}$. 
Another induction argument can now be employed to show that $1-\frac{\delta}{3}$ bounds the latter expression from below.

\section{Proof of Lemma~\ref{suc}} 
\label{L5}
Assume that at timestep $t$, $(s,a^1,a^2) \notin K^i_t$, $l^i(s,a)=0$ and $\mathsf{learn}^i(s,a)=\mathrm{true}$, and suppose that $m$ experiences of $(s,a^1,a^2)$ happen at timesteps $t\leq t_1<t_2<\cdots<t_{m}$. 
Since $l^i(s,a)=0$ and $\mathsf{learn}^i(s,a)=\mathrm{true}$, an attempted update will necessarily happen.
Let $r^i[1],r^i[2],\ldots,r^i[m]$ and $s[1],s[2],\ldots,s[m]$ be the rewards and next states observed for the $m$ experiences of $(s,a^1,a^2)$ for player $i$.
Then define the random variable $X \coloneq r^i[j]+ \gamma v^i_{t}(s[j])$, letting $j$ range in $\{1,\ldots, m\}$, and note that $0 \leq X \leq \frac{1}{1-\gamma}$.

A direct application of the Hoeffding inequality with the choice of $m$ as in \eqref{m} yields
\begin{equation*}
    \frac{1}{m} \left(\sum^{m}_{j=1} r^i[j] + \gamma v^i_{t}(s[j]) \right) - \mathbb{E} \big\{X\big\} < \epsilon_1
\end{equation*}
with probability $1- \frac{\delta}{6|S||A^1||A^2|(1+2\kappa)}$.
Note that there can be at most $2|S||A^1||A^2|(1+2\kappa)$ instances of such an event. 
Since Lemma~\ref{decreas} shows that all Nash Q-value estimates are decreasing during the execution of Delayed Nash Q-learning algorithm, for the condition on line $24$ of Algorithm~\ref{alg},  write:
\begin{align*}
    \label{lem6}
    Q^i_{t}(s,a^1,a^2)- \frac{1}{m} \left(\sum^{m}_{j=1} r^i[j] + \gamma v^i_{t_j}(s[j]) \right) 
     &\geq Q^i_t(s,a^1,a^2)- \frac{1}{m} \left(\sum^{m}_{j=1} r^i[j] + \gamma v^i_{t}(s[j]) \right) 
    \\ &> Q^i_t(s,a^1,a^2) - \mathbb{E} \big\{X\big\} - \epsilon_1 
    \enspace,
\end{align*}
and because $(s,a^1,a^2) \notin K^i_t$ meaning $Q^i_t(s,a^1,a^2)-\mathbb{E} \big\{X\big\} > 3\epsilon_1$,
\[    Q^i_t(s,a^1,a^2) - \mathbb{E} \big\{X\big\} - \epsilon_1 > 2\epsilon_1
\enspace,
\]
guaranteeing success for the update at timestep $t_{m}$.
Working similarly to the proof of Lemma~\ref{optimism}, one concludes that the successful update will occur with probability at least $1-\frac{\delta}{3}$.

\section{Proof of Lemma~\ref{escape1}} 
\label{L6}
Suppose an unsuccessful update of $Q^i(s,a^1,a^2)$ occurs at timestep $t$, and right after, at timestep $t+1$ you observe $\mathsf{learn}^i(s,a^1,a^2)=\mathrm{false}$.
Set up a contradiction argument: under those conditions, \emph{assume that $(s,a^1,a^2) \notin K^i_{t+1}$.}
Since the update at $t$ was unsuccessful, $K^i_{t+1} = K^i_t$, which also implies that $(s,a^1,a^2) \notin K^i_t$.
Now label the times of the most recent $m$ experiences of $(s,a^1,a^2)$ as $b^i(s,a^1,a^2) \coloneq t_1<t_2<\cdots<t_m=t$.
The contrapositive of the statement proved in Lemma~\ref{suc} states that given an unsuccessful update at $t$, it must be $(s,a^1,a^2) \in K^i_{t_1}$.
Since $(s,a^1,a^2) \notin K^i_t$, some other update must have happened between $t_1$ and $t$.
Denote the timestep of that update $t^\ast$, and note that $t^\ast \geq b^i(s,a^1,a^2)$.
Observe now that the condition $t_1 = b^i(s,a^1,a^2) \leq t^\ast$ would not have allowed the $\mathsf{learn}$ flag to be set to $\mathrm{false}$ (line $27$ of Algorithm~\ref{alg}).
Therefore, there is a contradiction. 
The assumption originally made is invalid, which means $(s,a^1,a^2) \in K^i_{t+1}$.

\section{Proof of Lemma~\ref{ebound}} 
\label{L7}
Fix a state-action profile $(s,a^1,a^2)$ and chose a player $i$. 
Show first that if $(s,a^1,a^2) \notin K^i_t$ is experienced at timestep $t$, then within at most $2m$ subsequent experiences of $(s,a^1,a^2)$, a successful update of $Q^i(s,a^1,a^2)$ must occur; to do so, follow this process:

For $(s,a^1,a^2) \notin K^i_t$, distinguish two possible cases at the beginning of timestep $t$: either $\mathsf{learn}^i(s,a^1,a^2) = \mathrm{false}$ or $\mathsf{learn}^i(s,a^1,a^2) = \mathrm{true}$.
\begin{compactenum}[(i)]
\item Consider first the case where $\mathsf{learn}^i(s,a^1,a^2) = \mathrm{false}$. Assume that the most recent attempted update of $Q^i(s,a^1,a^2)$ occurred at some timestep $t'$ which was unsuccessful and set the flag $\mathsf{learn}^i(s,a^1,a^2)$ to $\mathrm{false}$. Then, according to Lemma~\ref{escape1}, it will be $(s,a^1,a^2) \in K^i_{t'+1}$.
However, now it is $(s,a^1,a^2) \notin K^i_t$, which implies that some update must have occurred at some $t^\ast$ with $t'+1<t^\ast<t$. 
Thus at the beginning of timestep $t$, the flag $\mathsf{learn}^i(s,a^1,a^2)$ will set to $\mathrm{true}$  (line $15$ of Algorithm~\ref{alg}). 
At timestep $t$ all conditions of Lemma~\ref{suc} (i.e. $\mathsf{learn}^i(s,a^1,a^2)=\mathrm{true}$, $(s,a^1,a^2) \notin K^i_{t}$ and $l^i(s,a^1,a^2)=0$) are satisfied, and thus the attempted update of $Q^i(s,a^1,a^2)$ upon the $m^\mathrm{th}$ visit of $(s,a^1,a^2)$ after timestep $t$ will have to be successful.
\item Take now the case where $\mathsf{learn}^i(s,a^1,a^2) = \mathrm{true}$.
It is given that an attempted update for $Q^i(s,a^1,a^2)$ will occur in at most $m$ additional experiences of $(s,a^1,a^2)$, which can be assumed occurring at timesteps $t_1<\cdots<t_{m}$, and it is $t_1\le t \le t_{m}$.
Consider the two possibilities: $(s,a^1,a^2) \notin K^i_{t_1}$ or $(s,a^1,a^2) \in K^i_{t_1}$.
In the former case, Lemma~\ref{suc} indicates that the attempted update at $t_{m}$ will be successful.
In the latter case, given that $(s,a^1,a^2) \notin K^i_t$, some successful update at $t^\ast$ must have taken place between $t_1$ and $t$ (since $K^i_{t_1}\neq K^i_t$). 
If the attempted update at $t_{m}$ is unsuccessful, $\mathsf{learn}^i(s,a^1,a^2)$ remains $\mathrm{true}$ and at timestep $t_{m}+1$ it is $\mathsf{learn}^i(s,a^1,a^2)=\mathrm{true}$, $l^i(s,a^1,a^2)=0$ and $(s,a)\notin K^i_{t_{m}+1}$; this now triggers Lemma~\ref{suc}, which implies that the attempted update of $Q^i(s,a^1,a^2)$ upon the $m^\mathrm{th}$ visit of $(s,a^1,a^2)$ after timestep $t_{m}+1$ (within at most $2m$ more experiences of $(s,a^1,a^1)$ after $t$), will be successful.
\end{compactenum}

You have thus shown that after an event  $(s,a^1,a^2) \notin K^i_t$, an update for $Q^i(s,a^1,a^2)$ must occur within at most $2m$ more experiences of $(s,a^1,a^2)$. 
The proof of Lemma \ref{sbound} shows why the total number of successful updates for $(s,a^1,a^2)$ is bounded by $\frac{1}{(1-\gamma)\epsilon_1}$. 
Given this fact, the total number of timesteps with $(s,a^1,a^2) \notin K^i_t$ is bounded by $\frac{2m}{(1-\gamma)\epsilon_1}$.

Generalizing the above statement for all state-action profiles and all players, one concludes that the total number of escape events (timesteps $t$ with $(s_t,a^1_t,a^2_t) \notin K_t$) is bounded by $4m\kappa$.

\vskip 0.2in

\end{document}